\documentclass{amsart}

\usepackage{amsmath, amsthm,amssymb, amsfonts}
\usepackage{rotating}

\def\wt{{\operatorname{wt}}}
\def\WE{{\operatorname{W}}}
\def\CW{{\operatorname{CW}}}

\def\SE{{\operatorname{SE}}}

\def\EW{{\operatorname{EW}}}

\def\Z{{\mathbb Z}}

\def\C{{\mathbb C}}
\def\F{{\mathbb F}}

\newtheorem{thm}{Theorem}

\newtheorem{prop}[thm]{Proposition}
\newtheorem{cor}[thm]{Corollary}
\newtheorem{lemma}[thm]{Lemma}

\title[Higher Weight Enumerators]{MacWilliams Identities for $m$-tuple Weight Enumerators}

\author{Nathan Kaplan}
\address{Department of Mathematics, Yale University, New Haven, CT 06511}
\email{nathan.kaplan@yale.edu}
\date{\today}

\subjclass[2010]{94B05}

\keywords{MacWilliams Theorem, Linear Code, Weight Enumerator, Codes over Rings}

\begin{document}
\maketitle

\begin{abstract}
Since MacWilliams proved the original identity relating the Hamming weight enumerator of a linear code to the  weight enumerator of its dual code there have been many different generalizations, leading to the development of $m$-tuple support enumerators.  We prove a generalization of theorems of Britz and of Ray-Chaudhuri and Siap, which build on earlier work of Kl\o ve, Shiromoto, Wan, and others.  We then give illustrations of these $m$-tuple weight enumerators.
\end{abstract}

In a 1963 article \cite{MacW}, MacWilliams gave an identity relating the weight enumerator of a linear code to the weight enumerator of its dual code.  Several authors have generalized this work in a few different directions.  One type of generalization leads to weight enumerators in more than two variables, such as the Lee and complete weight enumerators, and to weight enumerators for codes defined over alphabets other than $\F_q$.  For example, a MacWilliams theorem for codes over Galois rings was given by Wan \cite{Wan}.  Another type of generalization considered by several authors is to adapt the notion of weight to consider more than one codeword at a time.  This leads to the generalized Hamming weights of Wei \cite{Wei}, and to the MacWilliams type results for $m$-tuple support enumerators of Kl\o ve \cite{Klove}, Shiromoto \cite{Shiromoto}, Simonis \cite{Simonis}, and Ray-Chaudhuri and Siap \cite{RCS1, RCS2}.  Barg \cite{Barg}, and later Britz \cite{Britz, BritzEx}, generalized some of these results and gave matroid-theoretic proofs.  Britz \cite{BritzNew} also recently described new and broad connections between weight enumerators and Tutte polynomials of matroids.

We prove a MacWilliams type result that implies the two main theorems of Britz \cite{Britz}, which concern support weight enumerators of codes and in turn imply the earlier results of Kl\o ve \cite{Klove}, Shiromoto \cite{Shiromoto}, and Barg \cite{Barg}.  Our result also implies the main theorems of Ray-Chaudhuri and Siap \cite{RCS1, RCS2} giving MacWilliams theorems for complete weight enumerators of an $m$-tuple of codes $C_1, C_2,\ldots, C_m$ that are not necessarily the same.  As in \cite{RCS2}, we phrase our results in terms of codes over Galois rings instead of restricting ourselves to codes over fields.  One key feature of our result is that not only can the codes $C_1, C_2,\ldots, C_m$ be distinct, but they do not necessarily have to be defined over the same ring, a generalization suggested in Siap's thesis \cite{SiapT}.  This is not the first MacWilliams theorem for $m$-tuples of codes defined over different alphabets.  In \cite{BritzEx}, Britz gives such a result for codes defined over finite fields that are not necessarily the same, but there is an additional constraint that the codes must have the same vector matroid.  This result is phrased in terms of code structure families, a direction we will not pursue here.

We then mention some of the ways in which $m$-tuple support enumerators are used in the theory of linear codes and give some applications.

\section{Statement of Results}
We first give the necessary definitions to state MacWilliams' original theorem \cite{MacW}.  Let $\F_q$ be a finite field of $q$ elements, $N$ a nonnegative integer, and $C \subseteq \F_q^N$ a linear code.  Let $|C|$ denote the number of codewords of $C$, and let $\langle a,b\rangle$ denote the usual pairing on $\F_q^N$.  The \emph{Hamming weight} of any $f\in \F_q^N$, denoted $\wt(f)$, is the number of nonzero coordinates of $f$.  We define the \emph{Hamming weight enumerator} of $C$,
\[ 
\WE_C(X,Y) = \sum_{c\in C} X^{N-\wt(c)} Y^{\wt(c)},
\]
a homogeneous polynomial of degree $N$.

\begin{thm}[MacWilliams]
Let $C \subseteq \F_q^N$ be a linear code and let $C^\perp$ be its dual code. Then
\[\WE_{C^\perp}(X,Y) = \frac{1}{|C|} \WE_C(X+(q-1)Y, X-Y).\]
\end{thm}

Many authors have considered not only the weights of individual codewords, but weights coming from $m$-tuples of codewords.  We give some terminology from \cite{Simonis}.  We will usually denote codewords with superscripts when we are considering more than one since we will use subscripts to denote the coordinates of a codeword.

Let $[N]$ denote $\{1,\ldots, N\}$.  For $v = (v_1,\ldots, v_N) \in \F_q^N$, we define the \emph{support} of $v$ by $S(v) = \{e \in [N]\ \mid \ v_e\neq 0\}$.  Note that $\wt(v) = |S(v)|$.  If we consider a codeword $c$ as a $1 \times N$ row vector then $\wt(c)$ is the number of nonzero columns of this matrix.  We define the \emph{weight}, sometimes called the \emph{effective length}, of an $m$-tuple of vectors $(v^1,\ldots, v^m) \in (\F_q^N)^m$ as the number of nonzero columns of the $m \times N$ matrix with rows $v^1,\ldots, v^m$.  This is the size of the union of the supports of $v^1,\ldots, v^m$.  For such an $m$-tuple $(v^1,\ldots, v^m)$ we define its \emph{support}, $S(v^1,\ldots, v^m) = \bigcup_{i=1}^m S(v^i)$.  For a subspace $V$ of $\F_q^N$ we define its support as $S(V) = \bigcup_{v\in V} S(v)$.  Note that $S(V)$ is the union of the supports of any set of vectors generating $V$.  We define the \emph{weight} of $V$ as the size of this support.

We begin with the simplest generalization of the Hamming weight enumerator that considers multiple codewords at the same time.  Let $C_1,\ldots, C_m$ be linear codes over $\F_q^N$ and let $\underline{C} = C_1\times \cdots \times C_m$.  We define the $m$-tuple Hamming weight enumerator by
\[
W^{[m]}_{C_1,\ldots, C_m}(X,Y) = \sum_{(c^1,\ldots, c^m) \in \underline{C}} f(c^1,\ldots, c^m),
\] 
where if the $m$-tuple of vectors $(c^1,\ldots, c^m)$ has effective length equal to $r$, then $f(c^1,\ldots, c^m) = X^{N-r} Y^r$.  The main result of this paper implies a version of the MacWilliams theorem for this weight enumerator.

We now give one of the main theorems of \cite{Britz}.  For consistency we state this as an identity involving homogeneous polynomials, which is different from, but equivalent to, the original presentation.  For $E \subseteq [N]$, let $A_{E}^{[m]}$ denote the number of ordered $m$-tuples of codewords in $C$ whose support is $E$.  We also define $2N$ variables $X_1,\ldots, X_N, Y_1,\ldots, Y_N$ that indicate whether a certain position is in the support of a given $m$-tuple of codewords.  We define the \emph{$m$-tuple support enumerator} of a linear code $C$ of length $N$ as
\begin{eqnarray*}
\SE^{[m]}_C(X_1,\ldots, X_N, Y_1,\ldots, Y_N) & = &   \sum_{E \subseteq [N]} A_{E}^{[m]} \prod_{i \in [N] \setminus E} X_i \prod_{j \in E} Y_j \\
& = & \sum_{(c^1,\ldots, c^m) \in C^m} H(c^1,\ldots, c^m),
\end{eqnarray*}
where $H(c^1,\ldots, c^m) = \prod_{P=1}^N H_P(c^1,\ldots, c^m)$, and
\[H_P(c^1,\ldots, c^m) = \begin{cases} X_P \text{ if } (c_P^1,\ldots, c_P^m) = (0,\ldots, 0) \\
Y_P \text{ otherwise} \end{cases}.\]

\begin{thm}[Britz]\label{Britz}
Let $C \subseteq \F_q^N$ be a linear code and let $C^\perp$ be its dual code. Then
\begin{eqnarray*}
& & \SE^{[m]}_{C^\perp}(X_1,\ldots, X_N, Y_1, \ldots Y_N)  =  \\  
& &  \frac{1}{|C|^m} \SE^{[m]}_C(X_1 +(q^m-1)Y_1,\ldots, X_N + (q^m-1) Y_N, X_1-Y_1,\ldots, X_N-Y_N).
\end{eqnarray*}
\end{thm}

In this theorem the supports of $m$-tuples of codewords of $C$ are related to the supports of $m$-tuples of codewords of $C^\perp$.  This support enumerator keeps track of the supports, not just their sizes.  However, given an $m$-tuple of codewords $c^1,\ldots, c^m$ written as an $m \times N$ matrix, this weight enumerator tells us only about the positions of the nonzero columns, not what these columns are.

We next give some terminology related to codes over Galois rings and complete weight enumerators necessary to state Theorem 2.4 of \cite{RCS2}, the other main result that we generalize.  A finite Galois ring $R$ of characteristic $p^e$ and cardinality $p^{et}$ is isomorphic to $\Z_{p^e}[\xi]$, where $\xi$ is a root of an irreducible monic polynomial of degree $t$ over $\Z_{p^e}$.  We write Galois rings in this form.  We note that if $e=1$ then this ring has no zero divisors and is isomorphic to the finite field $\F_{p^t}$.  A code $C$ of length $n$ over $R$ is a submodule of $R^N$ and its elements are codewords.  There is a pairing $\langle c_1, c_2 \rangle$ for elements of $R^N$ just as there is for elements of $\F_q^N$, and there is an analogous definition of $C^\perp$.  We note that $\left(C^\perp\right)^\perp = C$.

Every element of $R$ can be written in terms of a particularly nice basis.  Let $s = p^{et}-1$ and $\{0 = z_0, z_1,\ldots, z_s\}$ be some enumeration of the elements of $R$.  Any $\beta \in R$ can be written in a unique way as
\[
\beta = \gamma_0 + \gamma_1 \xi + \gamma_2 \xi^2 +\cdots + \gamma_{t-1} \xi^{t-1},
\]
with $\gamma_0, \ldots, \gamma_{t-1} \in \Z_{p^e}$.  We define a character $\chi:\ R\rightarrow \C^*$ by 
\begin{equation}\label{chi}
\chi\left(\gamma_0 + \gamma_1 \xi + \gamma_2 \xi^2 +\cdots + \gamma_{t-1} \xi^{t-1}\right) = \zeta^{\gamma_0},
\end{equation}
where $\zeta$ is a $p^e$-th complex root of unity.  We restrict to this class of Galois rings rather than the more general class of Frobenius rings, because in this setting the additive characters of $R$ can be understood in this very concrete way.  For MacWilliams theorems over more general finite rings, see \cite{Wood1, Wood2}.

We next define the complete weight enumerator of a linear code $C \subseteq R^N$.  We give definitions for codes over Galois rings which can easily be specialized to the case where $R=\F_q$ is a finite field.  The complete weight enumerator of a code $C \subseteq R^N$ is a homogeneous polynomial in $p^{et}$ variables, $X_{z_0}, X_{z_1},\ldots, X_{z_s}$, one for each element of $R$.  

For $c = (c_1,\ldots, c_N) \in R^N$, we define $F(c) = \prod_{j=1}^{N} F'(c_j)$, where $F'(c_j) = X_{z_i}$ if $c_j = z_i$.  So, $F(c) = \prod_{i=0}^{s} X_{z_i}^{a_i(c)}$, where $a_i(c)$ is the number of $j \in [N]$ such that $c_j = z_i$. The \emph{complete weight enumerator} of $C$ is
\[
\CW_C(X_{z_0},\ldots, X_{z_{s}}) = \sum_{c \in C} F(c).
\]

The following MacWilliams Theorem for the complete weight enumerator of a code over a Galois ring is proven by Wan \cite{Wan}.
\begin{thm}{\label{Wan}}
Let $C\subset R^N$ be a linear code and let $\chi$ be defined as in equation (\ref{chi}). Then
\begin{eqnarray*}
& & \CW_{C^\perp}(X_{z_0},\ldots, X_{z_{s}}) =\\
& & \frac{1}{|C|}  \CW_C\left( \sum_{i=0}^s \chi(z_0 z_i) X_{z_i}, \sum_{i=0}^s \chi(z_1 z_i) X_{z_i},\ldots, \sum_{i=0}^s \chi(z_s z_i) X_{z_i} \right).
\end{eqnarray*}
\end{thm}

We also define the $m$-tuple complete weight enumerator of $C_1, \ldots, C_m$ where each $C_i \subseteq R_i^N$ and $R_i$ is a Galois ring with elements $\{0 = z^i_0, z^i_1,\ldots, z^i_{s_i}\}$.  Let $\underline{C} = C_1 \times \cdots \times C_m$ and $\underline{R} = R_1^N \times \cdots \times R_m^N$.

Suppose $c^i \in C_i$ for each $i \in [m]$.  For any $m$-tuple $(c^1,\ldots, c^m)$, we consider the $m \times N$ matrix with rows $c^1,\ldots, c^m$.  We define one variable for each of the $\prod_{i=1}^m (s_i +1)$ column vectors and write them:
\[
X_{(z^1_0,z^2_0,\ldots, z^m_0)}, X_{(z^1_0,\ldots, z^{m-1}_0, z^m_1)}, \ldots, X_{(z^1_0,\ldots, z^{m-1}_0, z^m_{s_m})},X_{(z^1_0,z^2_0,\ldots, ,z^{m-1}_1, z^m_0)}, \ldots, X_{(z^1_{s_1},z^2_{s_2},\ldots,z^m_{s_m})}.
\]
When we have one variable for each possible $m$-tuple we always order them lexicographically.

Let $a_{(i_1,\ldots, i_m)}(c^1,\ldots, c^m)$ be the number of columns of this matrix that are equal to $(z^1_{i_1}, \ldots, z^m_{i_m})$.  For now, we are not concerned with the positions of the columns equal to a fixed $m$-tuple, only the number of such columns.  We define
\[
F(c^1,\ldots, c^m) = \prod X_{(z^1_{i_1},\ldots, z^m_{i_m})}^{a_{(i_1,\ldots, i_m)}(c^1,\ldots, c^m)},
\] 
where the product is taken over all $(i_1,\ldots, i_m)$ satisfying $0 \le i_j \le s_j$ for each $j \in [m]$.  As a product over coordinates this is equal to 
\[
\prod_{j=1}^N F'(c_j^1,\ldots, c_j^m) \text{, where }
F'(c_j^1,\ldots, c_j^m) = X_{(z^1_{i_1}, z^2_{i_2},\ldots, z^m_{i_m})},
\]
if $(c_j^1,\ldots, c_j^m) = (z^1_{i_1},\ldots, z^m_{i_m})$.  We now define the \emph{$m$-tuple complete weight enumerator} of $C_1,\ldots, C_m$ as
\[
\CW^{[m]}_{C_1,\ldots, C_m}(X_{(z^1_0,\ldots, z^m_0)},\ldots, X_{(z^1_{s_1}, \ldots, z^m_{s_m})}) = \sum_{(c^1,\ldots, c^m) \in \underline{C}} F(c^1,\ldots, c^m).
\]

An $m$-tuple MacWilliams theorem for complete weight enumerators of codes over the same Galois ring $R$ is the main result of \cite{RCS2}. 
\begin{thm}{\label{CWRCS2}}
Let $\chi$ be as defined in $(\ref{chi})$ and $C_1,\ldots, C_m$ be linear codes defined over $R^N$.  We have
\begin{eqnarray*}
& & \CW^{[m]}_{C^\perp_1,\ldots, C^\perp_m}(X_{(z_0,\ldots, z_0)},\ldots, X_{(z_{s_1}, \ldots, z_{s_m})}) =\\
& & \frac{1}{\prod_{i=1}^m |C_i|} \CW^{[m]}_{C_1,\ldots, C_m}(Y_{(z_0,\ldots, z_0)},\ldots, Y_{(z_{s_1}, \ldots, z_{s_m})}),
\end{eqnarray*}
where
\[
Y_{(z_{i_1},\ldots, z_{i_m})} =  
\sum_{j_1 =0}^{s_1} \cdots \sum_{j_m =0}^{s_m} \left(\prod_{k=1}^m \chi(z_{j_k} z_{i_k})\right) X_{(z_{j_1},\ldots, z_{j_m})}.
\]
\end{thm}

We now define a support analogue of the $m$-tuple complete weight enumerator of linear codes $C_1,\ldots, C_m$.  The idea is to consider all possible $m$-tuples of codewords and to keep track of which of the possible column vectors occurs in each of the $N$ positions.  This is a homogeneous polynomial in $N \left(\prod_{j=1}^m (s_j+1)\right)$ variables $X_{P, (z^1_{i_1},\ldots, z^m_{i_m})}$ where $P \in [N]$ and for each $j \in [m],\ 0 \le i_j \le s_j$.

Suppose $(c^1,\ldots, c^m) \in \underline{C}$ with $c^i = (c^i_1,\ldots, c^i_N)$. Consider the $m \times N$ matrix with rows $c^1,\ldots, c^m$.  Let 
\[
G(c^1,\ldots, c^m) = \prod_{P = 1}^N G_P(c^1_P,\ldots, c^m_P), 
\]
where we define
\[
G_P(c_P^1,\ldots, c_P^m) = X_{P,(z^1_{i_1},\ldots, z^m_{i_m})},
\]
for $(c^1_P,\ldots, c^m_P) = (z^1_{i_1},\ldots, z^m_{i_m})$.

We now define the \emph{$m$-tuple exact weight enumerator} of $C_1,\ldots, C_m$,
\begin{eqnarray*}
& & \EW^{[m]}_{C_1,\ldots, C_m}(X_{1,(z^1_0,\ldots, z^m_0)},\ldots, X_{1,(z^1_{s_1},\ldots, z^m_{s_m})},\ldots, X_{N,(z^1_0,\ldots, z^m_0)},\ldots, X_{N,(z^1_{s_1},\ldots, z^m_{s_m})})\\
& &  = \sum_{(c^1,\ldots, c^m) \in \underline{C}} G(c^1,\ldots, c^m).
\end{eqnarray*}

For $m=1$ this weight enumerator coincides with the exact weight enumerator in the book of MacWilliams and Sloane \cite{MS}.  We note that the $m$-tuple exact weight enumerator contains strictly more information than the $m$-tuple complete weight enumerator since it keeps track not only of how many times each of the $\prod_{i=1}^m (s_i + 1)$ possible columns occurs, but also in what positions they occur.  It is clear that this weight enumerator completely specifies the words of each code $C_1,\ldots, C_m$.

Our main result is the following generalization of Theorems \ref{Britz} and \ref{CWRCS2}.
\begin{thm}\label{CompW1}
Let $C_1,\ldots, C_m$ be linear codes of length $N$ over Galois rings $R_1,\ldots, R_m$, with dual codes $C_1^\perp,\ldots, C_m^\perp$.  For each $i\in [m]$, let $\chi_i$ be a character on $R_i$ defined as in (\ref{chi}). Then
\begin{eqnarray*}
& & \EW^{[m]}_{C^\perp_1,\ldots, C^\perp_m}(X_{1,(z^1_0,\ldots,z^m_0)},\ldots, X_{1,(z^1_{s_1},\ldots, z^m_{s_m})},\ldots, X_{N,(z^1_0,\ldots, z^m_0)},\ldots, X_{N,(z^1_{s_1},\ldots, z^m_{s_m})})  = \\ 
& & \frac{1}{\prod_{i = 1}^m |C_i|} \EW^{[m]}_{C_1,\ldots, C_m}(Y_{1,(z^1_0,\ldots,z^m_0)},\ldots, Y_{1,(z^1_{s_1},\ldots,z^m_{s_m})},\ldots,Y_{N,(z^1_0,\ldots,z^m_0)},\ldots, Y_{N,(z^1_{s_1},\ldots,z^m_{s_m})}), 
\end{eqnarray*}
where for each $(\alpha^1, \ldots, \alpha^m)\in \underline{R}$ and $\alpha_P = (\alpha_P^1, \ldots, \alpha_P^m)$,
\[
Y_{P,(\alpha_P^1,\ldots, \alpha_P^m)} 
=
\sum_{\beta = (\beta^1,\ldots, \beta^m) \in R_1\times \cdots \times R_m} \left(\prod_{k=1}^m \chi_k(\alpha^k_P  \beta^k)\right)  X_{P,(\beta^1,\ldots, \beta^m)}.
\]
\end{thm}

We use this result to give a proof of the following analogue for $m$-tuple Hamming weight enumerators, which also follows from Theorem 2.1 of \cite{RCS1}.

\begin{thm}\label{HamW}
Let $C_1,\ldots, C_m$ be linear codes over $\F_q^N$, with dual codes $C_1^\perp,\ldots, C_m^\perp$.  Then
\[\WE^{[m]}_{C_1^\perp,\ldots, C_m^\perp}(X,Y) = \frac{1}{\prod_{i = 1}^m |C_i|} \WE^{[m]}_{C_1,\ldots, C_m}(X+(q^m-1)Y, X-Y).\]
\end{thm}

This result allows one to compare the effective length of \hbox{$m$-tuples} of vectors drawn from different linear codes of the same length, and gives a generalization of an earlier result of Shiromoto \cite{Shiromoto} concerning the effective lengths of $m$-tuples of vectors from the same linear code $C$.

In the final part of the paper we discuss extensions to $r$-th support weight enumerators.  Wei \cite{Wei} first considered the $r$-th generalized Hamming Weight $d_r(C)$, which is the smallest effective length of an $r$-tuple of codewords of $C$ that generate an $r$-dimensional subcode of $C$.  Kl\o ve \cite{Klove} was the first to prove MacWilliams type relations for these effective length distributions.  We first define the $r$-th support weight distribution $\{A_i^{(r)} \mid\ i\ge 0\}$ of $C$ where $A_i^{(r)}$ is the number of $r$-dimensional subspaces of $C$ that have support of size exactly $i$.

We define the \emph{$r$-th support weight enumerator} of a linear code $C$,
\[
\WE_C^{(r)}(X,Y) = \sum_{i=0}^N A_i^{(r)} X^{N-i} Y^i.
\]

Britz \cite{Britz} gave a generalization of this weight enumerator that considers not only the dimension of the subcode but also which of the coordinates in $[N]$ lie in the support of the subcode.  We consider an analogue of this $r$-th support weight enumerator for linear codes of length $N,\ C_1,\ldots, C_m$, not necessarily equal, and see that things do not carry over so neatly in this setting.  We discuss this issue and give some applications of our results.

We can express an $m$-tuple of elements of $\F_q^N$ as the rows of an $m \times N$ matrix.  A column of this matrix gives an $m$-tuple $(\alpha_1,\ldots, \alpha_m) \in \F_q^m$.  If we choose a basis for $\F_{q^m}$ over $\F_q$, we can think of this $m$-tuple as an element of $\F_{q^m}$.  The resulting code over $\F_{q^m}$ is no longer linear since it is not closed under scalar multiplication by elements of $\F_{q^m} \setminus \F_q$, but it is $\F_q$-linear.  Codes of this type are often called \emph{additive} codes.  We can then think of Theorem \ref{HamW} as a kind of MacWilliams theorem for additive codes over $\F_{q^m}$.  We will not pursue this interpretation further here, but it may be useful in future work.  For more on MacWilliams Theorems for additive codes see \cite{Wood2}.

\section{The Proof of Theorem \ref{CompW1}}

We prove Theorem \ref{CompW1} on $m$-tuple exact weight enumerators using an argument similar in spirit to one of the original proofs of the MacWilliams identity \cite{MacW}.  Similar ideas have been used by Britz and others \cite{Britz, Elkies}.  The main difficulty in this argument is giving a careful definition of the Fourier transform along with the proper analogue of discrete Poisson summation.

We recall the function
\[ 
G(c^1,\ldots, c^m) = \prod_{P=1}^N G_P(c^1,\ldots, c^m),
\]
where $G_P$ is defined in the previous section.  This is a function from $\underline{C}$ to an algebra over $\C$.  Let $\underline{C}^\perp = C_1^\perp \times \cdots \times C_m^\perp$.  Recall that on each $R_i^N$ there is a pairing $\langle u^i, v^i\rangle = \sum_{j=1}^N u^i_j v^i_j$, where $u^i = (u^i_1,\ldots, u^i_N)$ and $v^i = (v^i_1,\ldots, v^i_N)$, and that equation (\ref{chi}) defines a character $\chi_i$ on each $R_i$.  We define the Fourier transform of $G$ by 
\[
\widehat{G}(u^1,\ldots, u^m) = \sum_{\underline{v} \in \underline{R}} \left(\prod_{i=1}^m  \chi_i(\langle u^i, v^i\rangle)\right) G(\underline{v}).
\]

We first recall a lemma from \cite{Wan}.
\begin{lemma}\label{Wchar}
Let $C\subset R^n$ be a linear code, $C^\perp$ its dual and $\chi$ be defined as in (\ref{chi}).  Then, for fixed $v\not\in C^\perp$,
\[
\sum_{u\in C} \chi(\langle u,v\rangle) = 0.
\]
\end{lemma}

One of the main tools in our proof is the following version of discrete Poisson summation.
\begin{lemma}{\label{DPS}}
We have
\[
\sum_{\underline{v} \in \underline{C}^\perp} G(\underline{v}) = \frac{1}{\prod_{i=1}^m |C_i|} \sum_{\underline{u} \in \underline{C}} \widehat{G}(\underline{u}).
\]
\end{lemma}

\begin{proof}
We consider
\[
\sum_{\underline{u} \in \underline{C}} \widehat{G}(\underline{u})  =  \sum_{\underline{u} \in \underline{C}} \sum_{\underline{v} \in \underline{R}} \left(\prod_{i=1}^m  \chi_i(\langle u^i, v^i\rangle)\right) G(\underline{v}).
\]

We express this double sum in two parts based on whether $\underline{v}$ is in $\underline{C}^\perp$ or not:
\begin{eqnarray*}
\sum_{\underline{u} \in \underline{C}} \sum_{\underline{v} \in \underline{R}} \left(\prod_{i=1}^m  \chi_i(\langle u^i, v^i\rangle)\right) G(\underline{v}) & = &  \sum_{\underline{u} \in \underline{C}} \sum_{\underline{v} \in \underline{C}^\perp} \left(\prod_{i=1}^m  \chi_i(\langle u^i, v^i\rangle)\right) G(\underline{v}) \\
& + & \sum_{\underline{u} \in \underline{C}} \sum_{\underline{v} \in \underline{R} \setminus \underline{C}^\perp} \left(\prod_{i=1}^m  \chi_i(\langle u^i, v^i\rangle)\right) G(\underline{v}).
\end{eqnarray*}

We switch the order of summation in each of these double sums and consider the second one.  Let $\underline{v} \in \underline{R} \setminus \underline{C}^\perp$ and consider
\[
\sum_{\underline{u} \in \underline{C}} \left(\prod_{i=1}^m  \chi_i(\langle u^i, v^i\rangle)\right) = \prod_{i=1}^m \sum_{u^i \in C_i} \chi_i(\langle u^i, v^i \rangle).
\]
By Lemma \ref{Wchar}, this is zero.  We now see that the first double sum is given by
\[
 \sum_{\underline{v} \in \underline{C}^\perp} \sum_{\underline{u} \in \underline{C}} \left(\prod_{i=1}^m  \chi_i(\langle u^i, v^i\rangle)\right) G(\underline{v}).
\]
We see that for a fixed $\underline{v} \in \underline{C}^\perp$,
\[
\sum_{\underline{u} \in \underline{C}} \left(\prod_{i=1}^m  \chi_i(\langle u^i, v^i\rangle)\right) = \prod_{i=1}^m |C_i|,
\]
completing the proof.

\end{proof}

We now give the proof of Theorem \ref{CompW1}.
\begin{proof}
We sum the function $G$ over all $(c^1,\ldots, c^m) \in \underline{C}$.  This gives
\[ 
\sum_{(c^1,\ldots, c^m) \in \underline{C}} G(c^1,\ldots, c^m) = \EW^{[m]}_{C_1,\ldots, C_m}(X_{1,(z^1_0,\ldots,z^m_0)},\ldots, X_{N,(z^1_{s_1},\ldots,z^m_{s_m})}).
\]
Lemma \ref{DPS} implies that this is equal to 
\[
\frac{1}{\prod_{i=1}^m |C_i^\perp|} \sum_{(d^1,\ldots, d^m) \in \underline{C}^\perp} \widehat{G}(d^1,\ldots, d^m).
\]

We now consider the coordinates of $\widehat{G}(d^1,\ldots, d^m)$ one at a time.  Note that 
\begin{eqnarray*}
\widehat{G}(d^1,\ldots, d^m)  & = & \sum_{(g^1,\ldots, g^m) \in \underline{R}} \prod_{i=1}^m \chi_i(\langle d^i, g^i \rangle) G(g^1,\ldots, g^m) \\
& = & \sum_{(g^1,\ldots, g^m) \in \underline{R}} \prod_{i=1}^m \prod_{P=1}^N  \chi_i(d_P^i  g_P^i) G_P(g^1_P,\ldots, g^m_P),
\end{eqnarray*}
where $g^i = (g^i_1,\ldots, g^i_N)$.

We can switch the order of the sum and product and still account for every $(g^1,\ldots, g^m) \in \underline{R}$ exactly once.  This sum is equal to
\[ 
\prod_{P=1}^N \sum_{(g_P^1,\ldots, g_P^m) \in R_1\times \cdots \times R_m} \prod_{i=1}^m \chi_i (d_P^i, g_P^i) G_P(g^1_P,\ldots, g^m_P).
\]
Let $g_P = (g^1_P,\ldots, g^m_P)$ and $d_P = (d^1_P,\ldots, d^m_P)$.  We can rewrite the previous sum as
\[ 
\prod_{P=1}^N \sum_{g_P \in R_1\times \cdots \times R_m} \left(\prod_{k=1}^m \chi_k( d^k_P, g^k_P)\right) X_{P,(g_P^1,\ldots, g_P^m)},
\]
which completes the proof.
\end{proof}

\section{Applications of Theorem \ref{CompW1} to Other Weight Enumerators}

In this section we deduce Theorem \ref{CWRCS2} and then Theorem \ref{Britz} from Theorem \ref{CompW1}, and then deduce Theorem \ref{HamW} from Theorem \ref{CWRCS2}.

\begin{proof}[Proof of Theorem \ref{CWRCS2}]
For all $P \in [N]$ and all $(i_1,\ldots, i_m)$ with $0\le i_j \le s_j$, set $X_{P,(z^1_{i_1},\ldots, z^m_{i_m})} = X_{(z^1_{i_1},\ldots, z^m_{i_m})}$.  By definition, for any fixed $(i_1,\ldots, i_m)$ the variables $Y_{P,(z^1_{i_1},\ldots, z^m_{i_m})}$ for $P\in [N]$ are all equal.  We also see that for each $P$ we have $G_P(c^1,\ldots, c^m) = X_{(c^1_P,\ldots, c^m_P)}$.  

Therefore 
\[
G(c^1,\ldots, c^m) = \prod X_{(z^1_{i_1},\ldots, z^m_{i_m})}^{a_{(i_1,\ldots, i_m)}(c^1,\ldots, c^m)},
\]  
where the product is taken over all $(i_1,\ldots, i_m)$ satisfying $0 \le i_j \le s_j$ for each $j \in [m]$. Taking the sum over all $m$-tuples $(c^1, \ldots, c^m) \in \underline{C}$ gives the weight enumerator
$\CW^{[m]}_{C_1,\ldots, C_m}(X_{(z^1_0,\ldots,z^m_0)},\ldots, X_{(z^1_{s_1},\ldots, z^m_{s_m})})$.  

The observation that $Y_{P,(z^1_{i_1},\ldots, z^m_{i_m})} = Y_{(z^1_{i_1},\ldots, z^m_{i_m})}$ for all $P$ gives an identity like Theorem \ref{CWRCS2}, except that the $m$ Galois rings $R_i$ can be distinct.  Specializing to the case where each $R_i$ is the same completes the proof.
\end{proof}

Before proving the next result, we first recall a lemma on sums of characters.
\begin{lemma}\label{Char2}
Suppose $\alpha = (\alpha_1,\ldots, \alpha_m) \in \F_q^m \setminus (0,\ldots, 0)$ and $\psi$ is a non-trivial additive character on $\F_q$.  Then 
\[
\sum_{\beta = (\beta_1,\ldots, \beta_m) \in \F_q^m \setminus (0,\ldots, 0)} \psi(\langle \alpha, \beta \rangle) = -1.
\]
\end{lemma}

\begin{proof}
The map $\beta  \rightarrow \psi(\langle \alpha, \beta \rangle)$ is a character on the finite additive group $\F_q^m$.  Therefore, the sum of this character over all $\beta$ vanishes unless it is the trivial character, which is the case if and only if $\alpha = (0,\ldots, 0)$.  We see that 
\[
\sum_{\beta = (\beta_1,\ldots, \beta_m) \in \F_q^m \setminus (0,\ldots, 0)} \prod_{i=1}^m \psi(\alpha_i \beta_i) = 0 - \prod_{i=1}^m \psi(0) = -1 .
\]
\end{proof}

\begin{proof}[Proof of Theorem \ref{Britz}]
We suppose that each $R_i$ is the same finite field $\F_q$, and that for each $i\in [m],\ C_i = C_1$.  For convenience we write $C := C_1$ and let $\{0=z_0,z_1,\ldots, z_{q-1}\}$ be some enumeration of the elements of $\F_q$.  Let $\psi$ be a non-trivial additive character on $\F_q$.  For each $P\in [1,N]$ set $X_{P,(z_0,\ldots,z_0)} = X_P$ and for all other $m$-tuples $(i_1,\ldots, i_m)$, set $X_{P,(z_{i_1},\ldots, z_{i_m})} = Y_{P}$.

First consider 
\[
Y_{P,(z_0,\ldots,z_0)} = \sum_{\beta = (\beta_1,\ldots, \beta_m) \in \F_q^m} X_{P,(\beta_1,\ldots, \beta_m)}.
\]
This is equal to $X_P + (q^m-1) Y_P$.  

Suppose $\alpha_P = (\alpha_P^1,\ldots, \alpha_P^m) \neq (0,\ldots, 0)$ and consider
\[
Y_{P,(\alpha_P^1,\ldots, \alpha_P^m)} = \sum_{\beta = (\beta_1,\ldots, \beta_m) \in \F_q^m} \psi(\langle \alpha_P, \beta \rangle) X_{P,(\beta_1,\ldots, \beta_m)}.
\]
In this case, the map that takes $\beta \in \F_q^m$ to $\psi(\langle \beta, \alpha_P \rangle)$ is a non-trivial character.  From the $\beta = (0,\ldots, 0)$ term we get $X_P$ and from the other terms we get 
\[
Y_P \sum_{\beta \neq (0,\ldots, 0)} \psi(\langle \alpha_P,\beta \rangle) = -Y_P,
\]
by the above lemma.  Therefore, $Y_{P,(\alpha_P^1,\ldots, \alpha_P^m)} = X_P-Y_P$.

Collecting terms completes the proof.
\end{proof}

Finally, we use the finite field version of Theorem \ref{CWRCS2} to prove Theorem \ref{HamW}. 
\begin{proof}[Proof of Theorem \ref{HamW}]
First suppose that the Galois ring $R$ is the finite field $\F_q$.  For an $m$-tuple $(i_1,\ldots, i_m)$ satisfying $0 \le i_1,\ldots, i_m \le q-1$ and $(i_1,\ldots, i_m) \neq (0,\ldots, 0)$ set $X_{(z_{i_1},\ldots, z_{i_m})}$ equal to $Y$, and set $X_{(z_0,\ldots,z_0)} = X$.  We note that 
\[
Y_{(z_0,\ldots,z_0)} = \sum_{(z_{i_1},\ldots,z_{i_m}) \in \F_q^m} X_{(z_{i_1},\ldots, z_{i_m})} = X + (q^m-1) Y.
\]
Consider $\alpha = (\alpha_1,\ldots, \alpha_m)\in \F_q^m$ with $\alpha\neq (0,\ldots, 0)$.  By Lemma \ref{Char2}, we have
\[
Y_{(\alpha_1,\ldots, \alpha_m)} = X + \sum_{\beta = (\beta_1,\ldots, \beta_m) \neq (0,\ldots, 0)} \psi(\langle \alpha, \beta\rangle)Y = X - Y.
\]
We note that $a_{(0,\ldots, 0)}(c^1,\ldots, c^m)$ counts the number of occurrences of the zero column in the matrix with rows $c^1,\ldots, c^m$.  Collecting terms completes the proof.

\end{proof}

\section{Support Weight Enumerators and Applications}

Several authors have studied weight enumerators from $m$-tuples of codewords from a single linear code $C$ where these $m$-tuples are grouped by the dimension of the subcode that they generate. The main fact that allows one to adapt the MacWilliams theorem for $m$-tuple support enumerators to give information about only $m$-tuples of codewords of $C$ that span a subspace of dimension $r$ is the following classical result.

\begin{prop}\label{CountB}
Let $D$ be an $r$-dimensional subspace of $\F_q^N$.  The number of ordered $m$-tuples of vectors $(d^1,\ldots, d^m) \in D^m$ that span $D$ is independent of $D$.  It is equal to $[m]_r := \prod_{i=0}^{r-1} (q^m-q^i)$.
\end{prop}

Let $C$ be a linear code of length $N$ and dimension $k$ over $\F_q$.  It is now an elementary observation that 
\[W^{[m]}_C(X,Y) = \sum_{r=0}^k [m]_r W^{(r)}_C(X,Y).\]
Applying the MacWilliams theorem to this weight enumerator gives the following result originally due to Kl\o ve \cite{Klove}.
\begin{prop}[Kl\o ve]\label{Klo}
Let $C$ be a linear code of length $N$ and dimension $k$ over $\F_q$.  Then for any $m\ge 1$,
\[ \sum_{r=0}^{N-k} [m]_r W^{(r)}_{C^\perp}(X,Y) = \frac{1}{q^{km}} \sum_{r=0}^k [m]_r W^{(r)}_{C}(X+(q^m-1)Y,X-Y).\]
\end{prop}

Adapting this result for $m$-tuples of words from different codes is not so straightforward.  Suppose we have linear codes $C_1,\ldots, C_m$ that are not necessarily the same and want to consider only $m$-tuples of codewords $(c^1,\ldots, c^m) \in \underline{C}$ that span a particular $r$-dimensional subspace $D$ of $\F_q^N$.  It is no longer the case that the number of $m$-tuples spanning $D$ depends only on $r$.  For example, if we choose a one-dimensional space $D$, the number of $m$-tuples spanning $D$ depends on the number of $C_i$ that contain $D$.  In general, for a particular space, in order to know the number of $m$-tuples of codewords that span it, we must know the dimension of the intersection of this space with each of the codes $C_i$.

We next consider one of the simplest examples with unequal codes.  We will see that the analogue of Proposition \ref{Klo} is much more complicated.  Let $C_1$ and $C_2$ be distinct linear codes over $\F_q$ of the same length $N$.  Suppose that $C_1$ has dimension $k,\ C_2$ has dimension $l$, and $C_1 \cap C_2$ has dimension $s$.  For each subspace of the code generated by $C_1$ and $C_2$ that is spanned by some pair $(c^1, c^2)$ with $c^1\in C_1$ and $c^2\in C_2$, we can ask for the number of such pairs of codewords that span this subspace.  We see that only the pair $((0,\ldots,0),(0,\ldots, 0))$ spans the zero-dimensional subspace consisting only of the zero codeword.

We first consider one-dimensional spaces.  Suppose we have a one-dimensional subspace of $C_1\cap C_2$.  By Proposition \ref{CountB}, this is generated by $[2]_1 = q^2-1$ pairs.  A one-dimensional subspace of $C_1$ that does not lie in $C_1\cap C_2$ must have a zero-dimensional intersection with it, so can only be generated by a pair of the form $(c^1,0)$ where $c^1$ lies in the subspace.  There are $q-1$ nonzero vectors in a one-dimensional subspace of $\F_q^N$.  A similar statement holds for one-dimensional subspaces of $C_2$ that do not lie in $C_1\cap C_2$.  Adding these up gives 
\[
(q-1)W^{(1)}_{C_1}(X,Y) + (q-1) W^{(1)}_{C_2} + (q-1)^2 W^{(1)}_{C_1\cap C_2}(X,Y),
\]
since we have taken $2(q-1)$ of the pairs of vectors generating subspaces in $C_1 \cap C_2$ and $q^2-1-2(q-1) = (q-1)^2$.

We next consider two-dimensional subspaces of the code generated by $C_1$ and $C_2$.  We note that $C_1 \setminus \{C_1\cap C_2\} = C_1 \setminus C_2$.

\begin{prop}
Let $C_1$ and $C_2$ be linear codes over $\F_q$ of length $N$ and dimensions $k$ and $l$, respectively.  Suppose that $C_1 \cap C_2$ has dimension $s$.  Then
\begin{eqnarray*}
W^{[2]}_{C_1,C_2}(X,Y) & = & X^N + (q-1)\left(W^{(1)}_{C_1}(X,Y) + W^{(1)}_{C_2}(X,Y)\right) \\
& + & (q-1)^2 W^{(1)}_{C_1\cap C_2}(X,Y) \\
& + & (q^2-1)(q^2-q) W^{(2)}_{C_1\cap C_2}(X,Y)  \\
& + &  q(q-1)^2 \left(W^{(2)}_{C_1\setminus C_2}(X,Y) + W^{(2)}_{C_2 \setminus C_1}(X,Y)\right) \\
& + &   (q-1)^2 W^{(2)}_{\langle C_1, C_2 \rangle \setminus \{C_1 \cup C_2\}}(X,Y),
 \end{eqnarray*}
where 
\[
W^{(2)}_{C_i \setminus {C_1\cap C_2}}(X,Y) = \sum_{j=0}^N A^{(2)}_j X^{N-j} Y^j,\]
 and $A_j^{(2)}$ denotes the number of two-dimensional subcodes of $C_i$ that have a one-dimensional intersection with ${C_1\cap C_2}$ and weight $j$, and 
 \[
 W^{(2)}_{\langle C_1, C_2 \rangle \setminus \{C_1 \cup C_2\}}(X,Y) = \sum_{j=0}^N B^{(2)}_j X^{N-j} Y^j,
 \]
 where $B^{(2)}_j$ denotes the number of two-dimensional subcodes of the code spanned by $C_1$ and $C_2$ but are not subcodes of either $C_1$ or $C_2$, that have weight $j$.
 \end{prop}

\begin{proof}
The number of pairs of vectors generating a two-dimensional subspace of ${C_1\cap C_2}$ is ${[2]_2 = (q^2-1)(q^2-q)}$.  The number of such subspaces is given by ${((q^s-1)(q^s-q))/((q^2-1)(q^2-q))}$.  We next consider two-dimensional subspaces of $C_1$ that are not contained in ${C_1\cap C_2}$.  If such a space can be generated by a pair $(c^1,c^2)$ then ${c^2\in {C_1\cap C_2}}$.  Given such a space, if we first choose $c^2$ there are $q^2-q$ choices for $c^1$, since the space contains $q^2$ total vectors. There are $(q^s-1)/(q-1)$ one-dimensional subspaces of ${C_1\cap C_2}$ and $((q^s-1)(q^s-q))/((q^2-1)(q^2-q))$ two-dimensional subspaces.  There are $((q^k-1)(q^k-q))/((q^2-1)(q^2-q))$ two-dimensional subspaces of $C_1$ each containing $(q^2-1)/(q-1) = q+1$ one-dimensional subspaces.  Therefore, there are 
\[
\frac{(q^k-1)(q^k-q)}{(q^2-1)(q^2-q)}\cdot \frac{(q+1)(q-1)}{q^k-1} = \frac{q^{k-1}-1}{q-1}
\] 
two-dimensional subspaces of $C_1$ containing a given one-dimensional subspace of ${C_1\cap C_2}$.  We see that $(q^{s-1}-1)/(q-1)$ of these are actually two-dimensional subspaces of ${C_1\cap C_2}$.  Therefore, we have 
\[
\frac{q^{k-1}-1 - (q^{s-1}-1)}{q-1}\cdot  \frac{q^s-1}{q-1} = \frac{(q^k-q^s)(q^s-1)}{q(q-1)^2}
\] 
two-dimensional subspaces of $C_1 \setminus C_2$ that can be generated by a pair $(c^1, c^2)$ with $c^1\in C_1,\ c^2 \in C_2$.  For each such space there are $(q^2-q)(q-1)$ pairs generating it, giving a total of $(q^k-q)(q^s-1)$ pairs generating such subspaces.  This is the same as the total number of pairs $c^1 \in C_1\setminus  C_2,\ c^2 \in {C_1\cap C_2}$, giving a useful verification of this count.  We similarly count $(q^l-q^s)(q^s-1)$ pairs of vectors that generate a two-dimensional subspace of $C_2 \setminus C_1$.

Using similar techniques we see that there are $(q^k-q^s)(q^l-q^s)/(q-1)^2$ subspaces of the code generated by $C_1$ and $C_2$ that have trivial intersection with $C_1\cap C_2$, and that each of these is generated by $(q-1)^2$ pairs $(c^1,c^2)$ with $c^i \in C_i$.  We omit the details.
\end{proof}

 We can now apply Theorem \ref{HamW} to this expression and see that this is equal to $\left(|C_1^{\perp}| |C_2^{\perp}|\right)^{-1}$ times the right hand side where each $C_i$ is replaced with $C_i^\perp,\ {C_1\cap C_2}$ is replaced with $C_1^\perp \cap C_2^\perp$ and $(X,Y)$ is replaced with $(X+(q^2-1)Y, X-Y)$.
 
 We give an example in order to make this more concrete.  We give binary codes of length $6,\ C_1$ and $C_2$ in terms of generator matrices,  
 \[C_1 = \begin{pmatrix} 
 1 & 1& 0 & 0 & 0 & 0 \\
 1 & 0& 1 & 0 & 0 & 0 \\
  1 & 1& 1 & 1 & 1 & 1 
  \end{pmatrix},\ \ \ C_2 = \begin{pmatrix}
  1 & 1 & 1 & 0 & 0 & 0 \\ 
    1 & 1& 1 & 1 & 1 & 1 
\end{pmatrix}.\]
We see that $ C_1 \cap C_2$ is the one-dimensional subspace generated by $(1,1,1,1,1,1)$, and that
 \[C_1^\perp = \begin{pmatrix} 
 0 & 0& 0 & 0 & 1 & 1 \\
 0 & 0& 0 & 1 & 0 & 1 \\
  1 & 1& 1 & 1 & 1 & 1 
  \end{pmatrix},\ \ \ C_2^\perp = \begin{pmatrix}
  1 & 1 & 0 & 0 & 0 & 0 \\ 
  1 & 0 & 1 & 0 & 0 & 0 \\ 
  0 & 0 & 0 & 1 & 0 & 1 \\ 
  0 & 0 & 0 & 0 & 1 & 1 
\end{pmatrix},\]
showing that $C_1$ is not self-dual, but is permutation equivalent to its dual.

We compute 
\begin{eqnarray*}
W_{C_1}^{(1)}(X,Y) & = & 3 X^4 Y^2 + 3 X^2 Y^4 + Y^6,\ W_{C_2}^{(1)}(X,Y)  = 2 X^3 Y^3 +  Y^6,\\
 W_{C_1\cap C_2}^{(1)}(X,Y) & = &  Y^6,\ W_{C_1\cap C_2}^{(2)}(X,Y)  =   0,\ W_{C_1 \setminus C_2}^{(2)}(X,Y)  = 3 Y^6,\\
W_{C_2 \setminus C_1}^{(2)}(X,Y)  & = &  Y^6,\ W^{(2)}_{\langle C_1, C_2 \rangle \setminus \{C_1 \cup C_2\}}(X,Y)  =  3 (X^3 Y^3 + X^2 Y^4 + X Y^5 + Y^6). 
\end{eqnarray*}

The above proposition now gives 
\[W^{[2]}_{C_1,C_2}(X,Y) = X^6 + 3 X^4 Y^2 + 5 X^3 Y^3 + 6 X^2 Y^4 + 3 X Y^5 + 14 Y^6.\]
Applying Theorem \ref{HamW} gives
\[W^{[2]}_{C_1^\perp,C_2^\perp}(X,Y) = X^6 + 12 X^4 Y^2 + 6 X^3 Y^3 + 39 X^2 Y^4 + 42 X Y^5 + 28 Y^6.\]
We can also see this by noting that $C_1^\perp \cap C_2^\perp = \begin{pmatrix} 0 & 0 & 0 & 0 & 1 & 1\\  0 & 0 & 0 & 1 & 0 & 1 \end{pmatrix}$, and performing an analysis similar to the one above.  We can compute each of the polynomials in the statement of the theorem, add them up with the proper constants and get $W^{[2]}_{C_1^\perp, C_2^\perp}(X,Y)$.

We state a corollary of Theorem \ref{HamW} separately.
\begin{cor}
Let $m\ge 1$ and $C$ be a linear code of length $N$ over $\F_q$.  Then
\[ W_{C,\ldots, C, C^\perp, \ldots, C^\perp}^{[2m]}(X,Y) = \frac{1}{q^{Nm}} W^{[2m]}_{C,\ldots, C, C^\perp, \ldots, C^\perp}(X+(q^m-1)Y, X-Y),\]
where $C$ and $C^\perp$ are each repeated $m$ times.
\end{cor}
A self-dual code $C$ must have its $m$-tuple weight enumerators invariant under certain transformations.  This is the main idea behind Gleason's theorem giving necessary conditions for the weight enumerators of self-dual codes \cite{HP, NRS}.  This corollary lets us produce polynomials that are invariant under the $m$-tuple analogue of the MacWilliams transformation, but are not necessarily the $m$-tuple weight enumerators of self-dual codes, in fact, are not necessarily the $m$-tuple weight enumerators of any single code $C$.

Let $C_3$ be the binary code with generator matrix  $\begin{pmatrix} 1 & 0 & 0 & 1 & 0 & 0 \\ 0 & 1 & 1 & 0 & 0 & 0 \\ 0 & 0 & 0 & 1 & 1 & 1 \end{pmatrix}$.  Then,
\begin{eqnarray*}
W^{[2]}_{C_3, C_3^\perp}(X,Y) & = &  X^6 + 5 X^4 Y^2 + 8 X^3 Y^3 +11 X^2 Y^4 + 24 X Y^5 + 15 Y^6 \\
& = & \frac{1}{2^6} W^{[2]}_{C_3, C_3^\perp}(X+ 3 Y,X-Y),
\end{eqnarray*}
but this cannot be the $2$-tuple weight enumerator of any code.  This is because for a binary code $C$,
\[W^{[2]}_C(X,Y) = \sum_{r=0}^2 [2]_r W^{(r)}_C(X,Y),\]
so for each $i\in [1,N]$ the $X^i Y^{N-i}$ coefficient must be divisible by $3$, but the $X^4 Y^2$ term has coefficient $5$.

Let $C_4$ have generator matrix $\begin{pmatrix} 1 & 1 & 0 & 0 & 0 & 0 \\ 1 & 1 & 1 & 1 & 1 & 1 \end{pmatrix}$.  This code has 
\[W^{[2]}_{C_4, C_4^\perp}(X,Y) = X^6 + 9 X^4 Y^2 + 27 X^4 Y^2 + 9 Y^6,\]
which is the $2$-tuple weight enumerator of the self dual code $C_5$ with generator matrix 
\[ \begin{pmatrix} 1 & 1 & 0 & 0 & 0 & 0 \\ 0 & 0 & 1 & 1& 0 & 0 \\ 1 & 1 & 1 & 1 & 1 & 1 \end{pmatrix}.\]  
We can also ask, given a polynomial that arises as $W^{[m]}_C(X,Y)$ for some $C$, whether we can characterize the $m$-tuples of codes $C_1,\ldots, C_m$ that give the same $m$-tuple weight enumerator.

We can ask questions of the following type.  Given $m$ and $q$, which homogeneous polynomials $W(X,Y)$ of degree $N$ are invariant under the transformation sending it to $q^{\frac{-Nm}{2}} W(X+(q^m-1) Y, X-Y)$?  This is asking for a kind of analogue of Gleason's theorem for these $m$-tuple weight enumerators.  For more information on this subject see the work of Nebe, Rains, and Sloane \cite{NRS}.  We know that there are polynomials invariant under this transformation that cannot be the $m$-tuple weight enumerator of any code, for example polynomials with multiple coefficients not divisible by $q^2-1$.  What further necessary conditions can we find for such an invariant polynomial to occur as the $m$-tuple weight enumerator of a code?  We would like to be able to use results of this type to aid in the classification of self-dual codes, and in more general classification problems.

We note that $C_5$ has the same weight enumerator as $C_1$, but that these two codes have different $2$-tuple weight enumerators.  This implies that the $m$-tuple weight enumerator of $C$ does not determine the $(m+1)$-tuple weight enumerator.  It is less clear whether it is possible for two codes to have the same $(m+1)$-tuple weight enumerators and different $m$-tuple weight enumerators.  Extensive computer search produced the following example (and many others).  Let $D_1$ be the binary code of length $12$ and dimension $6$ with generator matrix
\[\left( \begin{array}{rrrrrrrrrrrr}
1 & 0 & 0 & 0 & 0 & 0 & 0 & 0 & 1 &
1 & 0 & 0 \\
0 & 1 & 0 & 0 & 1 & 0 & 1 & 0 & 1 &
0 & 0 & 1 \\
0 & 0 & 1 & 0 & 0 & 0 & 1 & 0 & 1 &
0 & 1 & 0 \\
0 & 0 & 0 & 1 & 0 & 0 & 0 & 0 & 0 &
1 & 1 & 1 \\
0 & 0 & 0 & 0 & 0 & 1 & 1 & 0 & 0 &
1 & 0 & 0 \\
0 & 0 & 0 & 0 & 0 & 0 & 0 & 1 & 1 &
1 & 0 & 0
\end{array} \right),\]
and let $D_2$ be the binary code of length $12$ and dimension $6$ with generator matrix 
\[\left(\begin{array}{rrrrrrrrrrrr}
1 & 0 & 0 & 0 & 1 & 0 & 1 & 0 & 0 &
0 & 1 & 0 \\
0 & 1 & 0 & 0 & 0 & 0 & 1 & 1 & 0 &
0 & 1 & 0 \\
0 & 0 & 1 & 0 & 1 & 0 & 0 & 0 & 0 &
0 & 0 & 0 \\
0 & 0 & 0 & 1 & 0 & 0 & 0 & 0 & 0 &
1 & 0 & 1 \\
0 & 0 & 0 & 0 & 0 & 1 & 1 & 1 & 0 &
0 & 0 & 1 \\
0 & 0 & 0 & 0 & 0 & 0 & 0 & 0 & 1 &
0 & 1 & 1
\end{array}\right).\]
We compute that $D_1$ has Hamming weight enumerator 
\[ X^{12} + X^{10} Y^2 + 3 X^9 Y^3 + 6 X^8 Y^4 + 15 X^7 Y^5 + 14 X^6 Y^6 + 9 X^5 Y^7 + 9 X^4 Y^8 + 5 X^3 Y^9 + X^2 Y^{10},\]
and that
\begin{eqnarray*} 
& & W^{[2]}_{D_1}(X,Y) = X^{12} + 3 X^{10} Y^2 + 9 X^9 Y^3 + 24 X^8 Y^4 + 75 X^7 Y^5 + 162 X^6 Y^6\\
& & + 399 X^5 Y^7  + 771 X^4 Y^8 + 957 X^3 Y^9 + 975 X^2 Y^{10} + 576 X Y^{11} + 144 Y^{12}.
\end{eqnarray*}
We compute that $D_2$ has Hamming weight enumerator  
\[X^{12}+ X^{10} Y^2 + 3 X^9 Y^3 + 8 X^8 Y^4 + 11 X^7 Y^5 + 12 X^6 Y^6 + 17 X^5 Y^7 + 7 X^4 Y^8 + X^3 Y^9 + 3 X^2 Y^{10},\]
and the same $2$-tuple weight enumerator as $D_1$.  Therefore, $(m+1)$-tuple weight enumerators do not determine $m$-tuple weight enumerators.  This is related to recent work of Britz \cite{BritzNew}, in which he shows that for a $k$-dimensional linear code $C$ the collection of $m$-tuple weight enumerators for all $m$ satisfying $1\le m \le k$ is equivalent to the Tutte polynomial of the matroid associated to $C$.

\section{The Repetition Code and the Parity Check Code}

We end this paper with one more type of example.  Let $R$ be the repetition code of length $N$ defined over $\F_q$, that is, the one-dimensional code generated by $(1,1,\ldots, 1)$.  Then $R^\perp$ is the parity check code, which consists of all vectors of $\F_q^N,\ (c_1,\ldots, c_N)$ with $c_1 + \cdots + c_N = 0$ in $\F_q$.  Let $C_1,\ldots, C_m$ be linear codes of length $N$ over $\F_q$.  It is easy to see how to determine higher weight enumerators involving $R$, and less obvious how to determine weight enumerators involving $R^\perp$.  Theorem \ref{HamW} gives one way to solve this problem.

For any $m\ge 1$,
\[W^{[m+1]}_{C_1,\ldots, C_m, R}(X,Y) = W^{[m]}_{C_1,\ldots, C_m}(X,Y) +(q-1) \prod_{i=1}^m |C_i| Y^N,\]
since we can either choose the all zero codeword from $R$, giving the first term, or one of the $q-1$ words of weight $N$, giving the second term.  Similarly, we see that for any $m\ge 1$,
\[W^{[m+s]}_{C_1,\ldots, C_m, R,\ldots, R}(X,Y) = W^{[m]}_{C_1,\ldots, C_m}(X,Y) +(q^s-1) \prod_{i=1}^m |C_i| Y^N,\]
where $R$ is repeated $s$ times.  More generally, the same result holds if $R$ is any one-dimensional code over $\F_q$ generated by a vector with all nonzero coordinates.  This will be our assumption on $R$ from now on.

\begin{prop}{\label{Rep}}
Let $C_1,\ldots, C_m$ be linear codes of length $N$ over $\F_q$ and let $R$ be a one-dimensional code over $\F_q$ of length $N$ generated by $(c_1,\ldots, c_N)$, where each $c_i$ is nonzero.  Then
\begin{eqnarray*}
W^{[m+s+t]}_{C_1,\ldots, C_m, R,\ldots, R, R^\perp,\ldots, R^\perp}(X,Y) & = &  (q^s-1) q^{(N-1)t} \prod_{i=1}^m |C_i| Y^N + \frac{(q^t-1)}{q^t} (X-Y)^N \\
&+& \frac{1}{q^t} W^{[m]}_{C_1,\ldots, C_m}(X+(q^t-1)Y, q^t Y),
\end{eqnarray*}
where $R$ is repeated $s$ times and $R^\perp$ is repeated $t$ times.
\end{prop}

\begin{proof}
We consider $W^{[m+s+t]}_{C_1,\ldots, C_m, R,\ldots, R, R^\perp,\ldots, R^\perp}(X,Y)$, where $R$ is repeated $s$ times and $R^\perp$ is repeated $t$ times.  From the previous paragraph we have
\[ W^{[m+s+t]}_{C_1,\ldots, C_m, R,\ldots, R, R^\perp,\ldots, R^\perp}(X,Y) = W^{[m+t]}_{C_1,\ldots, C_m, R^\perp,\ldots, R^\perp}(X,Y) +(q^s-1) q^{(N-1)t} \prod_{i=1}^m |C_i| Y^N,\]
since $|R^\perp| = q^{(N-1)}$.
We apply Theorem \ref{HamW} and see that 
\begin{eqnarray*}
& & W^{[m+t]}_{C_1,\ldots, C_m, R^\perp,\ldots, R^\perp}(X,Y)  =  \frac{1}{q^t \prod_{i=1}^m |C_i^\perp|} W^{[m+t]}_{C_1^\perp,\ldots, C_m^\perp, R,\ldots, R}(X+(q^{m+t}-1)Y, X-Y)\\
& = &  \frac{1}{q^t \prod_{i=1}^m |C_i^\perp|} \big( W^{[m]}_{C_1^\perp,\ldots, C_m^\perp}(X+(q^{m+t}-1)Y, X-Y) + (q^{t}-1) \prod_{i=1}^m |C_i|^\perp (X-Y)^N\big).
\end{eqnarray*}
Applying Theorem \ref{HamW} one more time gives
\begin{eqnarray*}
& & \frac{W^{[m]}_{C_1^\perp,\ldots, C_m^\perp}(X+(q^{m+t}-1)Y, X-Y)}{q^t \prod_{i=1}^m |C_i^\perp|}  \\
&  = &  \frac{W^{[m]}_{C_1,\ldots, C_m}(X+(q^{m+t}-1)Y + (q^m-1)(X-Y), X+(q^{m+t}-1)Y -(X-Y))}{q^t \prod_{i=1}^m |C_i| |C_i|^\perp}  \\
& = & \frac{1}{q^t q^{Nm}} W^{[m]}_{C_1,\ldots, C_m}(q^m(X+(q^t-1)Y), q^m(q^t Y))\\
& = & \frac{1}{q^t} W^{[m]}_{C_1,\ldots, C_m}(X+(q^t-1)Y, q^t Y).
\end{eqnarray*}
\end{proof}

In certain cases it is not difficult to work out this proposition directly without use of the MacWilliams theorem and its generalizations.  For example this is not difficult when $m=1,\ q=2,\ s=0$, and $t=1$.  In this case $R^\perp$ is the even weight subcode of $\F_2^N$ and we get
\[W^{[2]}_{C_1, R^\perp}(X,Y) = \frac{(X-Y)^N}{2} + \frac{W_{C_1}(X+Y, 2Y)}{2}= W_{R^\perp}(X,Y) +  \frac{W_{C_1}^{(1)}(X+Y, 2Y)}{2},\]
since $W_{R^\perp}(X,Y) = \frac{ (X-Y)^N + (X+Y)^N}{2}$.  

Proposition \ref{Rep} gives a unified way to compute some of these more complicated higher weight enumerators.  Hopefully results of this type can be used to give further conditions on the existence of codes with certain weight enumerators or parameters.

\section{Acknowledgments}
The author thanks Professor Noam Elkies for introducing him to this area and for many helpful discussions.  He thanks Irfan Siap for bringing the references \cite{RCS1, RCS2, SiapT} to his attention. He also thanks Thomas Britz, Alexander Barg, Henry Cohn, and the two anonymous referees for very useful comments.  The author also thanks the National Science Foundation for supporting him with a Graduate Research Fellowship throughout part of this work.

{}

\end{document}